\documentclass[letterpaper, 10 pt, conference]{ieeeconf}  

\IEEEoverridecommandlockouts                              

\usepackage{graphicx}
\usepackage{subcaption}
\usepackage{comment}
\usepackage{epsfig} 
\usepackage{times} 
\usepackage{amsmath} 
\usepackage{amssymb}  
\usepackage{xcolor}
\usepackage{flushend}
\usepackage{booktabs}
\usepackage{cleveref}
\usepackage{multirow}
\usepackage{mathtools}

\newtheorem{mydef}{\bf Definition}
\newtheorem{mythm}{\bf Theorem}

\newtheorem{mylem}{\bf Lemma}
\newtheorem{mycol}{\bf Corollary}

\newtheorem{myass}{\bf Assumption}
\newtheorem{remark}{Remark}
\usepackage{algpseudocode}
\usepackage{algorithmicx,algorithm}

\newcounter{WSQcomment}

\title{\LARGE \bf
Information-Theoretic Adaptive Cooling for Deterministic MPPI via Entropy Feedback
}

\author{Shuqi Wang, Wenrong Sun, Tao Han, Yue Gao and Xiang Yin
\thanks{This work was supported by the National Natural Science 
Foundation of China (62173226, 62061136004).}
\thanks{S. Wang, T. Han, Y. Gao and X. Yin are with 
the School of Automation \& Intelligent Sensing, Shanghai Jiao Tong University, 
Shanghai 200240, China. 
e-mail: \{wangshuqi, coolhantao, yuegao, yinxiang\}@sjtu.edu.cn.  W. Sun is with the Department of Physics, 
The Hong Kong University of Science and Technology, Hong Kong S.A.R., China. 
e-mail: wsunam@connect.ust.hk}
}

\begin{document}

\maketitle
\thispagestyle{empty}
\pagestyle{empty}

\begin{abstract}
This paper investigates deterministic optimal control using Model Predictive Path Integral (MPPI) control, 
a sampling-based and derivative-free framework well suited for systems with complex dynamics and nonsmooth objectives. 
In deterministic MPPI, the temperature must be driven to zero to recover the true optimum, yet the design of an effective cooling schedule remains a fundamental challenge. Existing methods typically rely on predefined open-loop schedules, which limit the efficiency and robustness of the algorithm.
To overcome this limitation, we propose an Information-Theoretic Adaptive Cooling (ITAC) framework that uses the Shannon entropy of the importance weights as an online feedback signal to regulate the temperature. The proposed mechanism adapts the cooling rate to the current sampling state, enabling fast progress when the weights are diffuse and cautious cooling when they become concentrated. We prove asymptotic convergence of the resulting scheme to the deterministic optimum, and further derive a critical entropy threshold that leads to a smooth barrier against premature weight collapse. Experiments on nonsmooth signal temporal logic motion-planning tasks show that ITAC improves sampling efficiency and achieves substantially faster convergence than state-of-the-art baselines without sacrificing the derivative-free nature of MPPI.
\end{abstract}

\section{Introduction}
Model Predictive Path Integral (MPPI) control is a powerful sampling-based framework for solving complex optimal control problems \cite{williams2017model}. 
Rooted in path-integral stochastic optimal control theory \cite{kappen2005linear,theodorou2010generalized}, MPPI computes an optimal control distribution by minimizing the KL-divergence relative to a base distribution, thereby reformulating control as an expectation computation \cite{williamsInformationTheoreticModel2017a}. 
Its main advantages are that it is derivative-free and highly parallelizable: with GPU-accelerated trajectory sampling, MPPI can evaluate thousands of candidate trajectories in real time, making it effective for high-dimensional systems with complex dynamics and constraints \cite{williams2017model}. 
These strengths have enabled successful applications in autonomous driving \cite{williamsInformationTheoreticModel2017a}, agile quadrotor \cite{yi2024covompc}, robotic manipulation \cite{pezzato2025sampling}, and control under formal temporal-logic specifications \cite{zheng2026signal,halder2025trajectoryplanningsignaltemporal}; 
see \cite{kazim2024recent} for a recent survey.

While MPPI was originally developed for stochastic optimal control, recent studies have increasingly explored its use in deterministic systems. 
For example, \cite{homburger2025optimality} provided a theoretical justification for applying MPPI to deterministic optimal control by showing that, as the sampling distribution is progressively concentrated through joint annealing of the exploration covariance and the temperature, the MPPI solution approaches the true deterministic optimum.
Building on this perspective, \cite{halder2025trajectoryplanningsignaltemporal} further extended deterministic path integral optimization to trajectory planning under signal temporal logic specifications. 
More  recently,   \cite{fazlyab2026preconditioned} offered a variational optimization-theoretic perspective on MPPI, showing that the classical method can be interpreted as a preconditioned gradient descent step on a free-energy objective. 
 
Since MPPI relies on iterative sampling, its practical performance is inherently tied to sampling efficiency. 
A key factor governing this process is the \emph{temperature parameter} $\lambda$. 
In particular, to eliminate the intrinsic bias of MPPI and ensure convergence to the true deterministic optimum, the temperature must be driven to zero \cite{homburger2025optimality}. 
However, this requirement gives rise to a fundamental \emph{cooling-rate dilemma}. 
On the one hand, reducing $\lambda$ too aggressively causes the importance weights to become extremely skewed. As a result, the effective sample size (ESS) rapidly collapses, leading to a sharp increase in Monte Carlo estimation error. 
On the other hand, cooling too conservatively preserves sampling diversity but slows down convergence toward the optimum. 
In the language of nonequilibrium thermodynamics, the former phenomenon, where the system freezes prematurely into a suboptimal state before adequately exploring the solution space, is often referred to as \emph{glassification} \cite{kirkpatrick1983optimization}. 
In the sequential Monte Carlo literature, it is more commonly described as \emph{weight degeneracy/collapse} \cite{kappen2016adaptive}.

To mitigate the concentration of importance weights, recent advances have mainly followed two directions.  One line of work improves the sampling distribution itself, for example through adaptive importance sampling, ancillary controllers, or covariance shaping based on local quadratic approximations  \cite{asmar2023model,trevisan2024biased,schramm2026variancereduced,yi2024covompc}. Although effective at maintaining a high ESS, these methods often rely on gradient or Hessian information, thereby weakening MPPI's key advantage as a derivative-free method for handling highly nonsmooth dynamics and cost landscapes. Another line of work focuses on adaptive temperature tuning \cite{watson2023inferring}. While this approach highlights the practical importance of adaptive cooling, it does not guarantee that the temperature is driven to zero and treats the update rule as an inference-based heuristic without a principled connection to MPPI's underlying variational structure.

In this paper, we study MPPI for deterministic optimal control and develop a principled adaptive cooling mechanism for its sampling process. 
Unlike existing approaches that rely on predefined open-loop cooling schedules, the proposed Information-Theoretic Adaptive Cooling (ITAC) framework uses the real-time Shannon entropy of the importance weights as a feedback signal to regulate the temperature online according to the current sampling state. 
We prove that this mechanism guarantees asymptotic convergence to the deterministic optimum. 
Moreover, by deriving a critical entropy threshold as a necessary condition for Monte Carlo error bounds, we establish a theoretical barrier against premature glassification.  
Numerical experiments on nonsmooth signal temporal logic tasks further show that ITAC accelerates convergence by up to 73\% compared with state-of-the-art baselines, while preserving the derivative-free nature.

\section{Preliminaries}\label{sec:prelim}
In this section, we briefly review the standard MPPI method for deterministic optimal control. 
Consider the discrete-time deterministic control system of form
\begin{equation}\label{eq:dynamics}
    \mathbf{x}_{k+1} = f(\mathbf{x}_k, \mathbf{u}_k),
\end{equation}
where $\mathbf{x}_k \in \mathbb{R}^n$ and $\mathbf{u}_k \in \mathbb{R}^d$ denote the state and control input at time step $k$, respectively. 
For a given initial state $\mathbf{x}_0$ and control sequence 
$\mathbf{U}= (\mathbf{u}_0,\cdots,\mathbf{u}_{K-1})$, 
the resulting trajectory $\mathbf{X}=(\mathbf{x}_0,\cdots,\mathbf{x}_{K})$ is uniquely determined by \eqref{eq:dynamics}.

Let $\phi:\mathbb{R}^n\to \mathbb{R}$ be the terminal cost and 
$\mathcal{L}: \mathbb{R}^n\times \mathbb{R}^d\to \mathbb{R}$ be the stage cost. 
Given a planning horizon  $K$, our objective is to find an optimal control sequence 
$\mathbf{U}^*_{\mathrm{det}} = (\mathbf{u}_0^*,\dots,\mathbf{u}_{K-1}^*) \in \mathbb{R}^{d\times K}$ 
that minimizes the total cost, i.e.,
\begin{equation}\label{eq:cost}
    \mathbf{U}^*_{\mathrm{det}}
    =
    \arg\min_{\mathbf{U}}
    \left(
    \phi(\mathbf{x}_K) + \sum_{k=0}^{K-1} \mathcal{L}(\mathbf{x}_k,\mathbf{u}_k)
    \right).
\end{equation}
Furthermore, we assume that the running cost takes the form\vspace{-3pt}
\[
\mathcal{L}(\mathbf{x}_k, \mathbf{u}_k)
=
c(\mathbf{x}_k)
+\frac{1}{2}
\left(
\mathbf{u}_k^{\top}R\mathbf{u}_k
\right),
\]
where $c:\mathbb{R}^{n}\rightarrow\mathbb{R}_{\ge 0}$ is a  state-dependent running cost and the system input is penalized quadratically with $R\in\mathbb{R}^{d\times d}$.
Accordingly, given an initial state $\mathbf{x}_0$ and an input sequence $\mathbf{U}$ of length $K$, we define the state-dependent cost as\vspace{-3pt}
$$
S(\mathbf{x}_0, \mathbf{U})
=
\phi(\mathbf{x}_K)
+
\sum_{k=0}^{K-1} c(\mathbf{x}_k),
$$
which aggregates the running state costs and the terminal cost along the trajectory induced by~\eqref{eq:dynamics}.

Model predictive path integral  control is a sampling-based, derivative-free algorithm that solves~\eqref{eq:cost} iteratively. 
At a high level, it proceeds as the following steps.

\textbf{Step 1 (Initialization).}
Initialize a nominal control sequence
$\hat{\mathbf{U}} = (\hat{\mathbf{u}}_0, \dots, \hat{\mathbf{u}}_{K-1})$,
which can be chosen as the zero sequence or generated using heuristic methods.

\textbf{Step 2 (Sampling).}
Sample $M$ i.i.d.\ perturbation sequences
$\mathbf{V}^{(m)} = (\mathbf{v}_0^{(m)}, \dots, \mathbf{v}_{K-1}^{(m)})$
according to\vspace{-3pt}
\begin{equation}
    \label{eq:definition_p}
\mathbf{v}_k^{(m)} \sim \mathcal{N}(\mathbf{0},\, \Sigma),  \;\;   p(\mathbf{V} )
    :=
    \prod_{k=0}^{K-1} \mathcal{N}(\mathbf{v}_k;\, \mathbf{0},\, \Sigma),
\end{equation}
where the covariance matrix  $\Sigma$ is chosen such that $R=\lambda \Sigma^{-1}$ for some scalar $\lambda>0$, referred to as the \emph{temperature parameter}. 
For each sampled perturbation sequence, 
construct the  perturbed control sequence
$\mathbf{U}^{(m)}$ by
\[
\mathbf{u}_k^{(m)} = \hat{\mathbf{u}}_k + \mathbf{v}_k^{(m)},
\qquad k=0,\dots,K-1, m=1,\dots,M.
\]

\textbf{Step 3 (Weighting).}
Each perturbed control sequence $\mathbf{U}^{(m)}$ is rolled out through the dynamics~\eqref{eq:dynamics} to obtain its associated trajectory cost 
$S^{(m)} := S\bigl(\mathbf{x}_{0},\, \mathbf{U}^{(m)} \bigr)$. 
We  define\vspace{-3pt}
\[
J_{\mathbf{x}_0,\hat{\mathbf{U}}}(\mathbf{V})
:=
S(\mathbf{x}_0, \hat{\mathbf{U}}+\mathbf{V})
+
\lambda \sum_{k=0}^{K-1} \hat{\mathbf{u}}_k^\top \Sigma^{-1}\mathbf{v}_k.
\]
Since $\mathbf{x}_0$, $\hat{\mathbf{U}}$, $\lambda$, and $\Sigma$ are fixed within each iteration round of the weight update, we simplify the notation by writing
$J(\mathbf{V})$ for $J_{\mathbf{x}_0,\hat{\mathbf{U}}}(\mathbf{V})$, and define  $J_m := J(\mathbf{V}^{(m)})$. 
We define the normalized \emph{importance weight} associated with the $m$th sample by\vspace{-3pt}
\begin{equation}\label{eq:weight_J}
    w_m
    =
    \frac{\exp\bigl(-J_m/\lambda\bigr)}
    {\sum_{n=1}^{M}\exp\bigl(-J_n/\lambda\bigr)}.
\end{equation}

\textbf{Step 4 (Update).}
Compute the new nominal control sequence $\hat{\mathbf{U}}$ as the weighted average\vspace{-3pt}
\begin{equation}\label{eq:update}
    \hat{\mathbf{u}}_k \;\leftarrow\; \hat{\mathbf{u}}_k + \sum_{m=1}^{M} w_m\, \mathbf{v}_k^{(m)}, \quad k = 0,\dots,K-1.
\end{equation}

\textbf{Step 5 (Iteration).}
Update the temperature parameter $\lambda$ according to a prescribed \emph{cooling schedule}, and update the covariance $\Sigma$ accordingly so that $R=\lambda\Sigma^{-1}$. 
Then return to \textbf{Step~2} with the updated nominal control sequence $\hat{\mathbf{U}}$. 
This process is repeated until a convergence criterion is met. 
To asymptotically recover the true deterministic optimum $\mathbf{U}^*$, the temperature parameter must converge to zero as the iteration proceeds, i.e., $\lambda \to 0$~\cite{homburger2025optimality}.

\section{Adaptive Temperature Cooling}\label{sec:general_ideas}

In the MPPI framework, a key factor governing the sampling iterations is the temperature parameter $\lambda$. 
Let $j$ denote the iteration index, and let $\lambda_j$ denote the temperature parameter at iteration $j$. 
In existing works on MPPI for deterministic control \cite{homburger2025optimality,halder2025trajectoryplanningsignaltemporal,yoon2022samplingcomplexitypathintegral}, the cooling schedule is typically chosen as a fixed geometric decay, namely,
\[
\lambda_{j+1} = \nu^2 \lambda_j,
\]
for some constant $\nu \in (0,1)$. 
Although such a mechanism guarantees that the temperature converges to zero, the choice of the fixed decay rate $\nu$ remains largely heuristic: an overly aggressive schedule may cause premature collapse of the importance weights, whereas an overly conservative one may lead to unnecessarily slow progress. 
A more promising alternative is to adopt an adaptive cooling schedule that adjusts the decay rate dynamically according to the ``current status" of the sampling process.

In this work, we propose to use entropy as a feedback signal to adaptively regulate the temperature. 
Formally, for a probability distribution $q$, its Shannon entropy is defined by
\[
H(q)=\mathbb{E}_{q}[-\log q].
\]
Given $M$ samples with normalized importance weights $\{w_m\}_{m=1}^{M}$, we define the normalized sampling entropy as
\[
\hat{H}
=
\frac{-\sum_{m=1}^{M} w_m\log w_m}{\log M}.
\]
Based on this quantity, we consider a generic adaptive cooling schedule of the form
\[
\lambda_{j+1}=\Gamma(\hat{H}_j)\lambda_j,
\]
where $\Gamma:(0,1]\to(0,1)$ is a feedback function of the current normalized entropy. 
We choose $\Gamma(\cdot)$ to be decreasing in $\hat{H}$ so that the temperature is reduced more aggressively when $\hat{H}$ is close to $1$, namely when the importance weights remain relatively spread out, and more conservatively when $\hat{H}$ is close to $0$, namely when the weights have already become concentrated. 
In this way, the cooling schedule automatically adapts to the current sampling state, thereby balancing fast convergence and protection against premature weight collapse.

The remainder of this paper develops this idea in a rigorous manner. 
In the present section, we first explain why Shannon entropy serves as a natural feedback signal by examining the associated free-energy functional (Section~\ref{sec:why_entropy}). 
We then decompose the total control error to reveal the two competing objectives that any cooling schedule must balance (Section~\ref{sec:error_decomp}). 
Section~\ref{sec:MPPI_bias} introduces the class of admissible entropy-based cooling schedules and shows that they ensure the asymptotic convergence property $\lambda_j \to 0$. 
Section~\ref{sec:anti_glass} derives a critical entropy threshold $H_c$ from sampling-complexity theory and constructs a smooth barrier mechanism to mitigate premature glassification. 
Finally, Section~\ref{sec:experiments} validates the proposed framework on challenging motion-planning tasks.

\subsection{Why Choose Entropy as a Feedback Signal}\label{sec:why_entropy}
To understand why Shannon entropy provides an appropriate feedback signal, we briefly revisit the original derivation of the MPPI framework and its connection to deterministic optimization.
In essence, MPPI lifts the \emph{pointwise optimization} problem in \eqref{eq:cost} to a \emph{distributional optimization} problem.
Instead of directly searching for a single optimal control sequence $\mathbf{U}^*_{\mathrm{det}}$,   a nominal control sequence $\hat{\mathbf{U}}$ is introduced and one seeks an optimal distribution $q(\mathbf{V})$ over control perturbations $\mathbf{V} \in \mathbb{R}^{d \times K}$ that concentrates on low-cost regions.
Within this formulation, the optimal perturbation distribution is obtained by minimizing the \emph{free-energy functional} \cite{williamsInformationTheoreticModel2017a}:
\begin{equation}
    \zeta_\lambda(q) = \mathbb{E}_{q}[S(\mathbf{x}_0, \hat{\mathbf{U}} + \mathbf{V})] + \lambda D_{\mathrm{KL}}(q \| p), \label{eq:free_energy_functional}
\end{equation}
where $p$ is the base exploration distribution defined in \eqref{eq:definition_p}, and $\lambda > 0$ is the temperature parameter.
As shown in \cite{williamsInformationTheoreticModel2017a}, minimizing \eqref{eq:free_energy_functional} is equivalent to solving the original deterministic optimization problem in \eqref{eq:cost}.

A critical property of this formulation in our context is the structural invariance of the functional during the iterative optimization process. In MPPI, the control cost matrix $R = \lambda \Sigma^{-1}$ is fixed throughout the iterations, the quadratic control penalty embedded within the KL-divergence term $\lambda D_{\mathrm{KL}}(q \| p)$ is maintained. This ensures that the relative weighting between the state-dependent cost $\mathbb{E}_q[S]$ and the control effort penalty is preserved, even as the temperature $\lambda$ is adjusted. Consequently, the free energy functional \eqref{eq:free_energy_functional} remains consistent and stable for the deterministic objective \eqref{eq:cost} across the entire cooling schedule.

Since $D_{\mathrm{KL}}(q\|p) = \mathbb{E}_{q}[-\log p] - H(q)$, this functional admits the decomposition
\begin{equation}
\label{eq:free_energy_decomp}
    \zeta_\lambda(q)
    = \underbrace{\mathbb{E}_{q}[S]}_{\text{expected cost}}
    + \lambda\,\underbrace{\mathbb{E}_{q}[-\log p]}_{\text{reference mismatch}}
    - \lambda\,\underbrace{H(q)}_{\text{entropy reward}}.
\end{equation}
It reveals that $\lambda$ governs the balance between a reference mismatch penalty and an entropy reward.  Large $\lambda$ favors distributional diffusion and consistency with the base measure $p$ (exploration-dominated regime); small $\lambda$ shifts priority toward cost minimization and weight concentration (exploitation-dominated regime).  The cooling schedule, therefore, dictates when the system should transition from exploration to exploitation.

Eq.~\eqref{eq:free_energy_decomp} shows that the Shannon entropy $H(q)$ is an intrinsic quantity in the free-energy trade-off while its sample-based surrogate $\hat{H}$ serves as a practical feedback signal for modulating the cooling rate $\Gamma(\hat{H})$. A suitable feedback signal must reflect how diffuse the current distribution is. By giving $\hat{H}
\approx 1$ for nearly uniform sampling and large $\lambda$, more rapid decay for $\lambda$ is necessary for efficiency. While $\hat{H}\to 0$ when the weight is concentrated and small $\lambda$, cooling should decelerate, indicating that further aggressive cooling may lead to premature concentration. Therefore, the temperature controlling $\Gamma(\hat{H})\in [0,1]$ should act negatively with respect to $\hat{H}$ that $d\Gamma(\hat{H})/d\hat{H}<0$.

\subsection{Error Decomposition}
\label{sec:error_decomp}
 
A standard calculus-of-variations argument applied to Eq. \eqref{eq:free_energy_functional} (setting $\delta \zeta_\lambda / \delta q = 0$) yields the unique minimizer
\begin{equation}\label{eq:optimal_dist}
    q^*(\mathbf{V}) = \frac{1}{\eta}\exp\!\Bigl(-\frac{S(\mathbf{V})}{\lambda}\Bigr)\, p(\mathbf{V}),
\end{equation}
with $\eta$ the normalizing constant.  The finite-sample estimate Eq.~\eqref{eq:weight_J}--\eqref{eq:update} is precisely the self-normalized importance-sampling approximation of $\mathbb{E}_{q^*}[\mathbf{v}]$ using samples drawn from $q$.  Taking the expectation under $q^*$ defines the ideal \emph{infinite-sample MPPI solution} at temperature $\lambda$:
\begin{equation}\label{eq:inf_sample}
    \tilde{\mathbf{U}}^*_{\mathrm{MPPI}}(\lambda) = \mathbb{E}_{q^*}[\mathbf{V}] = \frac{\int \mathbf{V}\, \exp(-S(\mathbf{V})/\lambda)\, p(\mathbf{V})\,d\mathbf{V}}{\int \exp(-S(\mathbf{V})/\lambda)\, p(\mathbf{V})\,d\mathbf{V}}.
\end{equation}
This quantity bridges the finite-sample MPPI estimate $\hat{\mathbf{U}}$ and the true deterministic optimum $\mathbf{U}^*_{\mathrm{det}}$: it removes finite-sample noise from $\hat{\mathbf{U}}$ but still carries a temperature-dependent bias relative to $\mathbf{U}^*_{\mathrm{det}}$.  By the triangle inequality, the total control error decomposes as
\begin{equation}\label{eq:error_decomp_eq}
    \bigl\|\hat{\mathbf{U}} - \mathbf{U}^*_{\mathrm{det}}\bigr\|
    \;\leq\;
    \underbrace{\bigl\|\hat{\mathbf{U}} - \tilde{\mathbf{U}}^*_{\mathrm{MPPI}}\bigr\|}_{\text{MC estimation error}}
    \;+\;
    \underbrace{\bigl\|\tilde{\mathbf{U}}^*_{\mathrm{MPPI}} - \mathbf{U}^*_{\mathrm{det}}\bigr\|}_{\text{MPPI bias}}.
\end{equation}
The two terms are governed by fundamentally different mechanisms and impose opposing requirements on the temperature schedule:
 
\emph{MPPI Bias.}
Under the regularity conditions of~\cite{homburger2025optimality}, the gap between the infinite-sample PI solution and the deterministic optimum satisfies $\|\tilde{\mathbf{U}}^*_{\mathrm{MPPI}}(\beta) - \mathbf{U}^*_{\mathrm{det}}\| = \mathcal{O}(\beta^2)$ with $\beta = \sqrt{\lambda/\lambda_0}$.  Driving this term to zero therefore requires $\lambda_j \to 0$, i.e., the temperature must eventually be cooled to zero.  We formalize the admissible class of cooling schedules that guarantee this convergence in Section~\ref{sec:MPPI_bias}.
 
\emph{MC Estimation Error.}
Even if the MPPI bias vanishes, the finite-sample estimate $\hat{\mathbf{U}}$ can deviate significantly from $\tilde{\mathbf{U}}^*_{\mathrm{MPPI}}$ when the importance weights degenerate.  Bounding this error within a tolerance at prescribed risk requires the Effective Sample Size (ESS) to remain sufficiently large, which in turn demands that the Shannon entropy of the weight distribution does not fall below a critical threshold.  This is precisely the anti-glassification mechanism developed in Section~\ref{sec:anti_glass}.

\section{Admissible Entropy-Based Cooling Schedule}
\label{sec:MPPI_bias}

As motivated by the error decomposition in Eq.~\eqref{eq:error_decomp_eq}, driving the MPPI bias to zero requires a temperature schedule that ensures $\lambda_j \to 0$ as $j \to \infty$. Intuitively, one might satisfy this by requiring a global contraction bound $\sup_{\hat{H}\in[0,1]} \Gamma(\hat{H}) \le \overline{\Gamma} < 1$, which guarantees that the temperature vanishes at least geometrically. However, such a global restriction is unnecessarily stringent and excludes adaptive cooling policies where the contraction rate vanishes as the ensemble stabilizes (i.e., $\lim_{\hat{H} \to 0} \Gamma(\hat{H}) = 1$). 

To rigorously analyze these flexible schedules, we must distinguish between the functional properties of the cooling schedule $\Gamma(\hat{H})$ and the regularity of the cost landscape that determines the evolution of the empirical entropy $\hat{H}$.

\begin{mydef}[Admissible Cooling Function]
\label{def:admissible_cooling_func}
A continuous function $\Gamma: (0,1] \to (0,1)$ is called an \emph{admissible cooling function} if it satisfies the \textbf{pointwise contraction} property: $\Gamma(\hat{H}) < 1$ for all $\hat{H} \in (0,1]$. We denote the set of such functions by $\mathcal{C}_\Gamma$.
\end{mydef}

The definition above specifies when an entropy-based cooling law is contractive, but contractivity alone is not sufficient for asymptotic cooling. Indeed, if the empirical entropy were allowed to vanish at a non-zero temperature, the update factor could approach a degenerate regime without forcing further decay of $\lambda_j$. To exclude this possibility, we impose a mild boundedness assumption on the cost functional. This assumption yields a strictly positive lower bound on the empirical entropy whenever the temperature is bounded away from zero, which will be the key ingredient in the convergence argument below.

\begin{myass}
\label{assum:bounded_S}
    The cost functional $J(\mathbf{V})$ is bounded on the support of the sampling distribution $p$, such that its oscillation $\Delta_J := \sup J - \inf J$ is finite.
\end{myass}

Assumption~\ref{assum:bounded_S} is mild and holds in many practical robotics settings with compact state and input domains. The finiteness of the cost oscillation $\Delta_J$ provides quantitative control over the relative spread of the Gibbs weights, thereby excluding entropy collapse at any strictly positive temperature. This observation leads directly to a uniform positive lower bound for the empirical entropy.

\begin{mylem}[Positivity of Empirical Entropy]
\label{lem:entropy_positivity}
 For any fixed sample size $M > 1$ and temperature $\lambda \ge \epsilon > 0$, the normalized empirical entropy $\hat{H}$ satisfies $\hat{H} \ge \delta_\epsilon$, where:
\begin{equation}
\label{eq:delta_epsilon}
    \delta_\epsilon = \frac{\log(1 + (M-1)e^{-\Delta_J/\epsilon})}{\log M} > 0.
\end{equation}
\end{mylem}

\begin{proof}
Let $\kappa = e^{\Delta_J/\epsilon} > 1$. From the Gibbs weight in Eq. \eqref{eq:weight_J}, the ratio of any two weights is bounded by $w_m/w_n \le \exp(\Delta_J/\lambda) \le \kappa$ for $\lambda \ge \epsilon$. We seek a strictly positive lower bound for the Shannon entropy $H_M(w) = \sum_{m=1}^M w_m \log(1/w_m)$.

Let $w_{(1)}$ denote the maximum weight. Since $w_m \le w_{(1)}$ for all $m = 1, \dots, M$, we have $\log(1/w_m) \ge \log(1/w_{(1)})$. Thus:
$$
    H_M(w) = \sum_{m=1}^M w_m \log \frac{1}{w_m} \ge \sum_{m=1}^M w_m \log \frac{1}{w_{(1)}} = \log \frac{1}{w_{(1)}}.
$$
To bound $w_{(1)}$ from above, we use the ratio constraint $w_{(1)}/w_m \le \kappa$, which implies $w_m \ge w_{(1)}/\kappa$. Summing over all weights:
$$
    1 = w_{(1)} + \sum_{m \neq (1)} w_m \ge w_{(1)} + (M-1) \frac{w_{(1)}}{\kappa}.
$$
This yields $1/w_{(1)} \ge 1 + (M-1)/\kappa$. Substituting this into the entropy bound:
$$
    H_M(w) \ge \log \left( 1 + \frac{M-1}{\kappa} \right) = \log(1 + (M-1)e^{-\Delta_J/\epsilon}).
$$
Dividing by $\log M$ yields the normalized lower bound $\delta_\epsilon$ in Eq.~\eqref{eq:delta_epsilon}. Since $M > 1, \Delta_J > 0,$ and $\epsilon > 0$, we have $e^{-\Delta_J/\epsilon} \in (0, 1)$, ensuring $\delta_\epsilon > 0$ for all $\lambda \ge \epsilon$.
\end{proof}

With the positivity of entropy established for any non-zero temperature, we now prove that the adaptive cooling process converges to the deterministic limit.

\begin{mylem}[Convergence of Adaptive Cooling]
\label{lem:convergent_cooling}
Let $\Gamma \in \mathcal{C}_\Gamma$ and consider the multiplicative update $\lambda_{j+1}=\Gamma(\hat{H}_j)\,\lambda_j$ with $\lambda_0>0$. Then $\lambda_j \to 0$ as $j \to \infty$.
\end{mylem}

\begin{proof}
For any sample path $v$, the sequence $\{\lambda_j\}$ is strictly decreasing since $\Gamma(\hat{H}_j) < 1$ for all $\hat{H}_j > 0$. Since $\lambda_j$ is bounded below by $0$, it must converge to a limit $L \ge 0$. 

Suppose $L > 0$. By the \emph{non-vanishing drive} property, since $\lambda_j \ge L$, there exists a constant $\delta_L > 0$ such that $\hat{H}_j \ge \delta_L$ for all $j$. On the compact interval $[\delta_L, 1]$, the continuity of $\Gamma$ ensures that it attains a maximum value $\overline{\Gamma}_L := \max_{\hat{H} \in [\delta_L, 1]} \Gamma(\hat{H})$. Due to the \emph{pointwise contraction} property, $\overline{\Gamma}_L < 1$. 

It follows that $\lambda_{j+1} \le \overline{\Gamma}_L \lambda_j$ for all $j$, which implies $\lambda_j \le \lambda_0 (\overline{\Gamma}_L)^j \to 0$ as $j \to \infty$. This contradicts the assumption $L > 0$. Thus, $L$ must be $0$, and $\lambda_j \to 0$.
\end{proof}

\begin{mythm}
	\label{thm:asymptotic_optimality}
	Let $\Gamma\in\mathcal{C}_\Gamma$ and define the coupled updates
	$\lambda_{j+1}=\Gamma(\hat{H}_j)\,\lambda_j$,\;
	$\Sigma_{j+1}=\Gamma(\hat{H}_j)\,\Sigma_j$
	with $\lambda_0>0$, $\Sigma_0\succ 0$.
	Set $\beta_j:=\sqrt{\lambda_j/\lambda_0}$.
	Assume the cost functional $J$ satisfies the regularity conditions of {\rm\cite{homburger2025optimality}}. Then:
	\begin{enumerate}
		\item[\emph{(a)}] $\beta_j\to 0$ as $j\to\infty$ almost surely;
		\item[\emph{(b)}] $\|\tilde{\mathbf{U}}^*_{\mathrm{MPPI}}(\beta_j)-\mathbf{U}^*_{\mathrm{det}}\| = \mathcal{O}(\beta_j^2)\to 0$ (\cite{homburger2025optimality});
		\item[\emph{(c)}] At each finite iteration $j$, $\Sigma_j\succ 0$.
	\end{enumerate}
\end{mythm}

\begin{proof}
	Part~(a) follows directly from Lemma~\ref{lem:convergent_cooling}. 
	Part~(b) is a direct application of \cite[Theorem~1]{homburger2025optimality}.
	For part~(c), since $\Gamma(\hat{H}) > 0$ for all $\hat{H} \in (0,1]$, the product $\prod_{k=0}^{j-1} \Gamma(\hat{H}_k)$ is strictly positive for any finite $j$, ensuring the sampling covariance remains non-degenerate.
\end{proof}

\begin{remark}[From fixed to entropy-driven cooling]
\label{rem:fixed_vs_adaptive}
All prior deterministic MPPI formulations employ a constant geometric decay $\Gamma(\hat{H})\equiv\nu^2$ for some fixed $\nu\in(0,1)$ (\cite{homburger2025optimality,halder2025trajectoryplanningsignaltemporal,yoon2022samplingcomplexitypathintegral}).
Such constant schedules trivially belong to $\mathcal{C}_\Gamma$ but are entirely agnostic to the informational state of the trajectory ensemble.
In contrast, the entropy-driven cooling $\Gamma_{\mathrm{base}}(\hat{H})=\nu^2(1-\hat{H})$ considered in our experiments adapts the decay rate automatically. When the weight distribution is diffuse 
($\hat{H}$ large), $\Gamma_{\mathrm{base}}$ is small  and cooling is aggressive. As the distribution concentrates ($\hat{H}\to 0^+$), $\Gamma_{\mathrm{base}}$ approaches $\nu^2$ and cooling slows down.
Since $\Gamma_{\mathrm{base}}(1)=0$, this function does not strictly belong to $\mathcal{C}_\Gamma$. However, $\hat{H}=1$ (perfectly uniform weights) has probability zero under continuous cost distributions, and the composite schedule $\bar{\Gamma}$ of Section~\ref{sec:anti_glass} satisfies 
$\bar{\Gamma}\in\mathcal{C}_\Gamma$ on the entire domain $(0,1]$.
\end{remark}

\section{Anti-Glassification Mechanism}
\label{sec:anti_glass}
The previous section addressed the MPPI bias (second term in Eq.~\eqref{eq:error_decomp_eq}) by establishing that any admissible schedule $\Gamma\in\mathcal{C}_\Gamma$ drives
$\lambda_j\to 0$ and thus eliminates the bias asymptotically.
We now turn to the MC estimation error (first term in Eq.~\eqref{eq:error_decomp_eq}), which imposes an opposing constraint: if $\lambda$ decreases too aggressively, the
importance weights degenerate, the Effective Sample Size collapses, and the finite-sample estimate $\hat{\mathbf{U}}$ diverges from $\tilde{\mathbf{U}}^*_{\mathrm{MPPI}}$.
This premature localization is analogous to glassy freezing in non-equilibrium thermodynamics \cite{kirkpatrick1983optimization}.
To quantify the point at which this failure occurs, we derive a critical entropy threshold $H_c$ below which the sampling complexity bounds of \cite{yoon2022samplingcomplexitypathintegral} are violated, and then introduce a smooth barrier function that prevents the cooling schedule from pushing $H$ below $H_c$.

Our goal is to derive a critical entropy threshold $H_c$
below which the MC error bounds are violated, and then use
it to construct a smooth barrier that prevents over-cooling.
The argument proceeds in four steps: 
(i)~recall the MC error bound and its dependence on
$\lambda$ and the sample size $M$
(Lemma~\ref{lem:mc_error});
(ii)~characterize how the ESS depends on the temperature
$\lambda$ and sample size $M$, providing a lower bound
(Corollaries~\ref{col:ess_bound_no_limit}--\ref{col:ess_bound});
(iii)~connect ESS to the Shannon entropy via the
R\'{e}nyi entropy (Lemma~\ref{lem:entropy_ess});
(iv)~combine these to obtain the threshold $H_c$
(Theorem~\ref{thm:critical_entropy}).

We begin by stating the
finite-sample error bound from
\cite{yoon2022samplingcomplexitypathintegral}, which
quantifies how the MC estimation error depends on the temperature $\lambda$, the sample size $M$, and
the prescribed tolerance $(\epsilon_1,\rho_1)$, $(\epsilon_2,\rho_2)$.
\begin{myass}[\cite{yoon2022samplingcomplexitypathintegral}]
    We assume that the trajectory cost $J(\mathbf{V})$ utilized in the Gibbs weights is non-negative. This assumption holds because the self-normalized importance sampling (SNIS) scheme is invariant to constant shifts in the cost functional. For any constant $C \in \mathbb{R}$, the normalized weight $w_m$ remains unchanged:
\begin{equation}
    w_m = \frac{\exp(-(J_m + C)/\lambda)}{\sum_{n=1}^M \exp(-(J_n + C)/\lambda)} = \frac{\exp(-J_m/\lambda)}{\sum_{n=1}^M \exp(-J_n/\lambda)}.
\end{equation}
In practice, one can always ensure $J \geq 0$ by shifting the costs relative to the minimum sampled value (i.e., $J \leftarrow (J - \min\{J_m\}_{m=1}^M)$) or adding a sufficiently large constant. This numerical shift does not alter the resulting control update $\hat{u}$ but allows the application of concentration inequalities that require non-negative exponents.
    \label{ass:non_negative_cost}
\end{myass}

\begin{mylem}[\cite{yoon2022samplingcomplexitypathintegral}, Proposition 1]
\label{lem:mc_error}
Under Assumption~\ref{ass:non_negative_cost}, the multiplicative and additive error between the true deterministic optimal control $[\mathbf{u}_k^*]_i$ and the Monte Carlo (MC) estimated control $[\hat{\mathbf{u}}_k^*]_i$ in each dimension $i$ is bounded by:
$$(1-\frac{\epsilon_1}{\mathbb{E}[w]})([\mathbf{u}_k^*]_i-\epsilon_2) \le [\hat{\mathbf{u}}_k^*]_i \le (1+\frac{\epsilon_1}{\mathbb{E}[w]})([\mathbf{u}_k^*]_i+\epsilon_2)$$
This bound holds with a risk probability $\rho_2$ associated with the additive error $\epsilon_2$, defined as:
$$\rho_1 := 2e^{-M\epsilon_1^2},\quad \rho_2 := \frac{1+\sqrt{2}}{M\epsilon_2^2}e^{\frac{2\mathbb{E}[J_m]}{\lambda}}$$
\end{mylem}

The key observation is that the risk $\rho_2$ couples $M$ and the exponential factor $e^{2\mathbb{E}[J_m]/\lambda}$. As $\lambda$ decreases, $\rho_2$ grows exponentially unless $M$ is increased accordingly. And the expected trajectory cost $\mathbb{E}[J_m]$, which itself depends on the cost landscape and is generally difficult to characterize \emph{a priori}.
Rather than tracking these quantities separately, we
next show that the Effective Sample Size (ESS) provides a single
quantity that governs this trade-off.
\begin{mydef}[Effective Sample Size]
    For $M$ unnormalized weights $\mathbf{w}_m = e^{-J_m/\lambda}$, the Effective Sample Size is defined as
$\mathrm{ESS} = {\bigl(\sum_{m=1}^{M} \mathbf{w}_m\bigr)^2}/{\sum_{m=1}^{M} \mathbf{w}_m^2}.$
\end{mydef}

We first work in the infinite-sample regime, where the sample averages can be replaced by their expectations. This allows a clean algebraic reduction that absorbs all quantities into a single ESS lower bound naturally. 
We then lift the result to the finite-sample setting by
adding a Hoeffding concentration layer, showing that the
same bound holds with high probability at a negligible
additional risk.
\begin{mycol}[ESS Lower Bound under Infinite-Sample]
\label{col:ess_bound_no_limit}
For a system evaluated with $M$ sample trajectories and an expected trajectory cost $\mathbb{E}[J_m]$, the ESS of the Boltzmann-weighted trajectories scales according to the temperature $\lambda$ satisfying the following relationship:
$$ESS \ge M e^{-\frac{2\mathbb{E}[J_m]}{\lambda}}$$
\end{mycol}
\begin{proof}
We can rewrite this in terms of sample expectations under infinite-sample: $ESS = M \frac{(\frac{1}{M}\sum \mathbf{w}_m)^2}{\frac{1}{M}\sum \mathbf{w}_m^2} \approx M \frac{\mathbb{E}[\mathbf{w}]^2}{\mathbb{E}[\mathbf{w}^2]}$. 
Because the costs are non-negative ($J_m \ge 0$), seen in Assumption~\ref{ass:non_negative_cost}, the weights are bounded $\mathbf{w}_i \in [0,1]$. Consequently, $\mathbf{w}_i^2 \le \mathbf{w}_i$, which implies that the second moment is bounded by the first moment: $\mathbb{E}[\mathbf{w}^2] \le \mathbb{E}[\mathbf{w}] \le 1$ (Assumption~\ref{ass:non_negative_cost}). 
Replacing the denominator with its upper bound of $1$, we obtain a conservative lower bound $ESS \ge M \mathbb{E}[\mathbf{w}]^2$.
By applying Jensen's inequality to the exponential weight function, we know that $\mathbb{E}[e^{-J_m/\lambda}] \ge e^{-\mathbb{E}[J_m]/\lambda}$ \cite{yoon2022samplingcomplexitypathintegral}. Squaring both sides and substituting this into our ESS inequality yields the lemma.
\end{proof}

In practice we have finitely many samples, so the bound above holds only approximately. The following corollary shows that it holds with high probability, with a concentration risk that is negligible
compared to the MC error risk $\rho_2$.
\begin{mycol}[ESS Lower Bound under Finite-sample]
\label{col:ess_bound}
Let the trajectory costs $J_1, \dots, J_M$ be independent and identically distributed, satisfying Assumption~\ref{ass:non_negative_cost}.
Then for any confidence level $\rho_{\mathrm{ess}} \in (0,1)$, with probability at least $1 - \rho_{\mathrm{ess}}$, the Effective Sample Size of the Boltzmann-weighted trajectories satisfies $\mathrm{ESS} \!\;\ge\;\! M\,e^{-\frac{2\mathbb{E}[J_m]}{\lambda}}$, 
where
\begin{equation}\label{eq:rho_ess}
    \rho_{\mathrm{ess}} \;\le\; \exp\!\left(-2M\left(e^{-\mathbb{E}[J_m]/\lambda} - e^{-2\mathbb{E}[J_m]/\lambda}\right)^2\right).
\end{equation}
\end{mycol}

\begin{proof}
Under Assumption~\ref{ass:non_negative_cost}, $J_m \ge 0$ implies $\mathbf{w}_m \in (0,1]$, hence $\mathbf{w}_m^2 \le \mathbf{w}_m$ for every $m$. Summing over all samples yields $\sum_{m} \mathbf{w}_m^2 \le \sum_{m} \mathbf{w}_m$, and therefore $\mathrm{ESS} \;\ge\; \sum_{m=1}^{M} \mathbf{w}_m.$
This inequality is exact and holds for every finite $M$ without any approximation.

Since the costs $J_1, \dots, J_M$ are i.i.d., the weights $\mathbf{w}_1, \dots, \mathbf{w}_M$ are also i.i.d.\ with $\mathbf{w}_m \in [0,1]$. By Hoeffding's inequality,
\begin{equation}\label{eq:hoeffding}
    \mathbb{P}\!\left(\sum_{m=1}^{M} \mathbf{w}_m < M\,\mathbb{E}[\mathbf{w}] - t\right) \;\le\; \exp\!\left(-\frac{2t^2}{M}\right).
\end{equation}
We wish to establish the event $\sum \mathbf{w}_m \ge M\,e^{-2\mathbb{E}[J_m]/\lambda}$, so we set
$t \;=\; M\!\left(\mathbb{E}[\mathbf{w}] - e^{-2\mathbb{E}[J_m]/\lambda}\right).$
By Jensen's inequality applied to the convex function $x \mapsto e^{-x/\lambda}$, we have $\mathbb{E}[\mathbf{w}] = \mathbb{E}\!\left[e^{-J_m/\lambda}\right] \ge e^{-\mathbb{E}[J_m]/\lambda},$
and since $\mathbb{E}[J_m] > 0$ implies $e^{-\mathbb{E}[J_m]/\lambda} > e^{-2\mathbb{E}[J_m]/\lambda}$, the offset $t$ is strictly positive. Substituting into~\eqref{eq:hoeffding}:
$$\mathbb{P}\!\left(\!\sum_{m=1}^{M}\!\! \mathbf{w}_m \!\!<\! M e^{-2\mathbb{E}[J_m]/\lambda}\!\!\right) \!\!\le\! \exp\!\left(\!\!-2M\!\!\left(\!\mathbb{E}[\mathbf{w}] \!\!-\! e^{-2\mathbb{E}[J_m]/\lambda}\!\right)^2\!\right)\!.$$
A  relaxation using $\mathbb{E}[\mathbf{w}] \ge e^{-\mathbb{E}[J_m]/\lambda}$ gives the explicit bound~\eqref{eq:rho_ess}.
Therefore, with probability at least $1 - \rho_{\mathrm{ess}}$, we have 
$\mathrm{ESS} \;{\ge}\; \sum_{m=1}^{M} \mathbf{w}_m \;\ge\; M\,e^{-2\mathbb{E}[J_m]/\lambda}$.
\end{proof}
\begin{remark}
The concentration risk $\rho_{\mathrm{ess}}$ decays exponentially in $M$ (cf.~\eqref{eq:rho_ess}), whereas the MC error risk $\rho_2$ in Lemma~\ref{lem:mc_error} decays only as $\mathcal{O}(M^{-1})$. Consequently, for any practical sample size, $\rho_{\mathrm{ess}} \ll \rho_2$, and the additional risk is negligible.
\end{remark}

\begin{mylem}
\label{lem:entropy_ess}
Given $M$ normalized importance weights $w_m = \mathbf{w}_m / \sum_{n} \mathbf{w}_n$, define the empirical Shannon entropy $H_M := -\sum_{m=1}^{M} w_m \log w_m \in [0, \log M]$. Then the Effective Sample Size satisfies
$$
\log(\mathrm{ESS}) \le H_M.
$$
\end{mylem}
\begin{proof}
By definition, $\mathrm{ESS} = 1/\sum_m w_m^2$, so
$\log(\mathrm{ESS}) = -\log\!\bigl(\sum_{m=1}^M w_m^2\bigr)$.
Since $(w_m)$ is a probability vector, this is precisely the R\'{e}nyi entropy of order $2$:
$$
H_{\mathrm{R\acute{e}nyi}}(\alpha) = \frac{1}{1-\alpha}\log\!\Bigl(\sum_{m=1}^M w_m^\alpha\Bigr), \;\;\log(\mathrm{ESS}) = H_{\mathrm{R\acute{e}nyi}}(2).
$$
By the log-sum inequality, $H_{\mathrm{R\acute{e}nyi}}(\alpha)$ is monotonically non-increasing in $\alpha$ for any discrete probability vector, and $\lim_{\alpha\to 1} H_{\mathrm{R\acute{e}nyi}}(\alpha) = H_M$. Therefore $\log(\mathrm{ESS}) = H_{\mathrm{R\acute{e}nyi}}(2) \le H_{\mathrm{R\acute{e}nyi}}(1) = H_M$.
\end{proof}

\begin{mythm}[Critical Entropy Threshold]
\label{thm:critical_entropy}
Under Assumption~\ref{ass:non_negative_cost}, define the critical entropy threshold
$$H_c := \log\!\left( \frac{1+\sqrt{2}}{\rho_2 \epsilon_2^2} \right).$$
Then $H_M \ge H_c$ is a \emph{necessary} condition for the MC error bounds of Lemma~\ref{lem:mc_error} to hold at risk $\rho_2$ and additive precision $\epsilon_2$.
Equivalently, if $H_M < H_c$, the ESS is guaranteed to fall below the sampling complexity requirement, and the error bounds are necessarily violated.
\end{mythm}

\begin{proof}
From Lemma~\ref{lem:mc_error}, the risk probability satisfies
$\rho_2 = \frac{1+\sqrt{2}}{M\epsilon_2^2}\,e^{{2\mathbb{E}[J_m]}/{\lambda}}$.
Rearranging yields
$    M\,e^{-{2\mathbb{E}[J_m]}/{\lambda}} \;=\; \frac{1+\sqrt{2}}{\rho_2\,\epsilon_2^2}.$
By Corollary~\ref{col:ess_bound}, $\mathrm{ESS} \ge M\,e^{-{2\mathbb{E}[J_m]}/{\lambda}}$, so the error bounds require $\mathrm{ESS} \;\ge\; \frac{1+\sqrt{2}}{\rho_2\,\epsilon_2^2}$.
Now suppose $H_M < H_c$. By Lemma~\ref{lem:entropy_ess}, $\log(\mathrm{ESS}) \le H$, so
$$\log(\mathrm{ESS}) \;\le\; H_M \;<\; H_c \;=\; \log\!\left(\frac{1+\sqrt{2}}{\rho_2\,\epsilon_2^2}\right),$$
which cause a contradiction that $\mathrm{ESS} < \frac{1+\sqrt{2}}{\rho_2\,\epsilon_2^2}$.
Therefore $H_M \ge H_c$ is necessary for the error bounds to hold.
\end{proof}

Theorem~\ref{thm:critical_entropy} provides a principled value for the entropy floor. If $H_M$ drops below $H_c$, the ESS is too small to guarantee the MC error bounds in a probabilistic sense.
To enforce this constraint, we define a smooth barrier function $\Gamma(\cdot):(0,1]\to(0,1)$ that wraps any base cooling function $\Gamma_{\mathrm{base}}\in\mathcal{C}_\Gamma$:
$$\Gamma(\hat{H}) = \Gamma_{\mathrm{protect}} + \left[ \Gamma_{\mathrm{base}}(\hat{H}) - \Gamma_{\mathrm{protect}} \right] \cdot \sigma\!\left( \kappa (\hat{H} - \hat{H}_c) \right)$$
where $\sigma(x)=1/(1+\exp(-x))$ is the standard sigmoid, $\kappa>0$ controls the steepness of the transition, $\hat{H}=H_M/\log M$, 
$\hat{H}_c=H_c/\log M$ is the normalized critical threshold, and $\Gamma_{\mathrm{protect}}\in(0,1)$ is a conservative fallback decay rate.
When $\hat{H}_j \gg \hat{H}_c$, the sigmoid saturates at~$1$ and $\Gamma_j\approx \Gamma_{\mathrm{base}}(\hat{H}_j)$, recovering the original cooling behavior.
When $\hat{H}_j$ drops to or below $\hat{H}_c$, the sigmoid transitions smoothly toward~$0$ and $\Gamma_j\to \Gamma_{\mathrm{protect}}$, overriding the base schedule with a conservative decay rate. Crucially, $\Gamma_{\mathrm{protect}} \in (0,1)$ still satisfies $\Gamma_{\mathrm{protect}} < 1$, so the temperature continues to decrease and the asymptotic guarantee $\lambda_j \to 0$ of Theorem~\ref{thm:asymptotic_optimality} is preserved. The barrier does not halt cooling; rather, it decelerates the descent so that the system passes through the low-entropy regime gradually, allowing the sampling distribution to readjust before further temperature reduction.
The temperature and covariance are then updated multiplicatively as $\lambda_{j+1} = \Gamma_j \lambda_j, \quad \Sigma_{j+1} = \Gamma_j \Sigma_j$.

\section{Experiments and Analysis}
\label{sec:experiments}

\subsection{Empirical Validation of the Critical Entropy Threshold}
\begin{figure}[t]
  \centering
  \includegraphics[width=0.8\columnwidth]{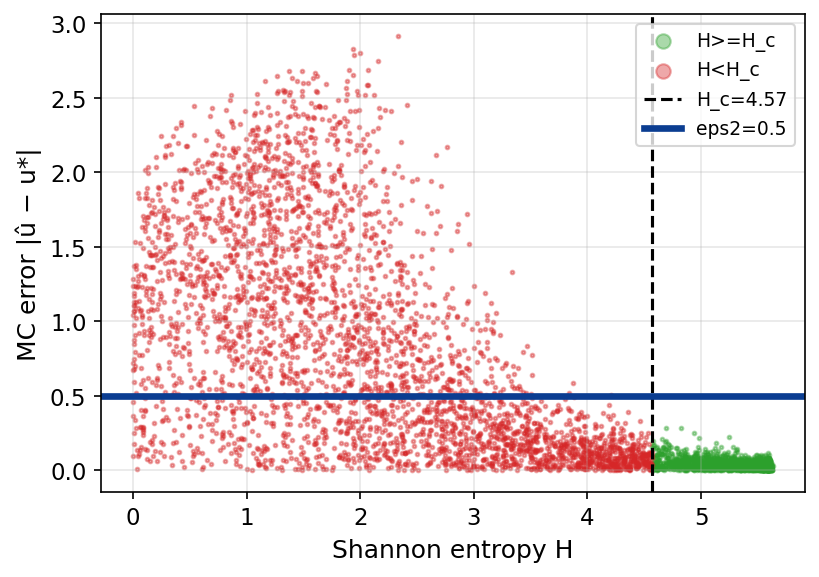}
\caption{MC estimation error versus the empirical Shannon weight entropy $H_M$, with the critical threshold $H_c$ and the error tolerance $\epsilon_2$ indicated by dashed/solid lines. Points are colored by whether $H_M \ge H_c$ or $H_M < H_c$. As guaranteed by Theorem~\ref{thm:critical_entropy}, the region $H_M < H_c$ exhibits frequent and severe violations of the $\epsilon_2$ tolerance, confirming that breaching the critical threshold leads to a breakdown of the MC error bounds. Interestingly, in the complementary region $H_M \ge H_c$, the empirical error also remains well-contained beneath the tolerance line. This result supports that the Shannon entropy serves as an effective real-time monitor for sampling quality.}

  \label{fig:hc_validation}
\end{figure}

Theorem~\ref{thm:critical_entropy} guarantees that the critical entropy threshold bounds the Monte Carlo (MC) estimation error in a probabilistic sense. To empirically validate this, we construct a one-dimensional importance-sampling setting (Fig.~\ref{fig:hc_validation}). 
Using $M=300$, $\epsilon_2=0.5$, and $\rho_2=10\%$, the theorem threshold is $H_c=\log\!\big((1+\sqrt{2})/(\rho_2\epsilon_2^2)\big)\approx 4.57 < \log M\approx 5.70$. 

Theorem~\ref{thm:critical_entropy} establishes that the Shannon entropy $H$ must satisfy $H \ge H_c(\rho_2, \varepsilon_2)$
to guarantee that the Monte Carlo error bounds of MPPI hold at risk $\rho_2$ and additive precision $\varepsilon_2$.  This threshold governs the anti-glassification barrier.

Choosing $H_c$ via the theorem ties the hyperparameter directly to the desired error tolerance. Tighter precision requires a higher entropy floor and therefore a more conservative cooling schedule.

Having validated the critical entropy threshold for the anti-glassification, we now evaluate the complete adaptive cooling schedule.
In all subsequent experiments, we instantiate the base cooling function as $\Gamma_{\mathrm{base}}(\hat{H})=\nu^2 (1-\hat{H})$, a valid member of $\mathcal{C}_\Gamma$ (cf.\ Remark~\ref{rem:fixed_vs_adaptive}), combined with the smooth barrier of Section~\ref{sec:anti_glass}.
We refer to this complete scheme as \textbf{ITAC MPPI}. For comparison, we also evaluate ITAC w.o. AG (our adaptive schedule without the anti-glassification barrier) and Vanilla MPPI (the standard baseline with a fixed cooling rate $\nu^2$).

\subsection{Case I: Reach-Avoid Tasks}
\label{sec:reach_avoid}
We first evaluate the adaptive cooling schedule on a large suite of
simple reach-avoid problems to assess its general applicability.
We randomly generate 50 scenarios, each with 1--5
circular obstacles. The objective is to reach the goal while avoiding all obstacles.
ITAC MPPI, ITAC without Anti-glassification, and vanila MPPI (\cite{williams2017model}) use identical initial temperature $\lambda_0$, base decay rate $\nu^2$, and sample size $M$; the only difference is the cooling rule.

\begin{table}[t]
\centering
\caption{Reach-avoid results over 50 randomly generated
scenarios. All methods use identical initial temperature $\lambda_0$, base decay $\nu^2$ and sample size $M$.}
\label{tab:reach_avoid}
\renewcommand{\arraystretch}{1.15}
\begin{tabular}{lcccc}
\toprule
Method & SR (\%) & Cost & Time (s) & Iters \\
\midrule
ITAC MPPI
  & \textbf{100.0}
  & \textbf{272.7\,$\pm$\,38.8}
  & \textbf{0.012\,$\pm$\,0.009}
  & \textbf{4.4\,$\pm$\,2.9} \\
  
ITAC w.o. AG
  & 100.0
  & 277.4\,$\pm$\,42.8
  & 0.012\,$\pm$\,0.009
  & 4.5\,$\pm$\,3.1 \\
Vanilla MPPI
  & 99.4
  & 322.2\,$\pm$\,68.3
  & 0.016\,$\pm$\,0.014
  & 5.7\,$\pm$\,4.8 \\
\bottomrule
\end{tabular}
\end{table}

Table~\ref{tab:reach_avoid} summarizes the results.
Both ITAC MPPI and ITAC w.o.\ AG achieve a perfect 100\% success rate across all 50 scenarios, compared to 99.4\% for Vanilla MPPI.
The full adaptive schedule produces the lowest cost (a $15.4\%$ reduction over Vanilla MPPI) with the smallest variance.
The ablation confirms that the entropy-driven cooling itself is the primary source of improvement, while the barrier contributes a further modest gain.
The result shows that a fixed decay rate that works well for one obstacle layout may be too aggressive or too conservative for another, whereas the entropy-driven cooling self-adjusts to the difficulty of each instance on standard planning problems.
We next stress-test it on significantly harder tasks involving non-smooth STL cost landscapes.

\vspace{-2pt}
\subsection{Case II: Harder Non-Smooth STL Tasks}
\vspace{-2pt}
We further evaluate ITAC MPPI and its ablation (ITAC w.o.\ AG) on three STL planning benchmarks of increasing complexity against the standard MPPI baseline (with fixed temperature schedule \cite{halder2025trajectoryplanningsignaltemporal}). The system is a 2D point-mass ($\mathbf{x}_{t+1} = A\mathbf{x}_t + B\mathbf{u}_t$, $dt=0.5$\,s) starting at $\mathbf{x}_0 = [1,1,0,0]^\top$. 
The tasks are: \textbf{Task A} (reach goal avoiding one obstacle); \textbf{Task B} (navigate past three obstacles, $N{=}3000$ rollouts); and \textbf{Task C} (pass through a narrow corridor gap before reaching a distant goal). We run 100 independent trials per task and report the STL success rate (SR), median iterations to convergence, and mean solve time. All methods use identical initial temperature $\lambda_0$, base decay $\nu^2$,
and sample size $M$. Table~\ref{tab:results} summarizes the performance across all tasks.
\vspace{0pt}

\begin{table}[t]
\centering
\caption{Comparative results (100 runs per method per task). SR: STL success rate. Med.\ Iters: median iterations to convergence. Time: mean solve time (s). Lower iters and lower time are better for equal SR.}
\label{tab:results}
\renewcommand{\arraystretch}{1.1}
\resizebox{\columnwidth}{!}{%
\begin{tabular}{llcccc}
\toprule
Task & Method & SR (\%) & Med.\ Iters & Time (s) & $\bar\rho$ \\
\midrule
\multirow{2}{*}{A}
  & ITAC MPPI & 100 & 38 & \textbf{2.22}          & 0.194 \\
  & ITAC w.o. AG        & 100          & 38          & 2.42          & 0.190 \\
  & Vanilla MPPI        & 100          & 50          & 4.03          & 0.192 \\
\midrule
\multirow{2}{*}{B}
  & ITAC MPPI & 100 & 13 & 0.83 & 0.559 \\
    & ITAC w.o. AG        & 100          & \textbf{12}          & \textbf{0.67}          & 0.566 \\
  & Vanilla MPPI        & 100          & 48          & 3.44          & 0.579 \\
\midrule
\multirow{2}{*}{C}
  & ITAC MPPI & \textbf{22}  & 201         & 15.28         & $-$0.037 \\
    & ITAC w.o. AG        & 18          & 201          & \textbf{10.88}          & $-$0.051 \\
  & Vanilla MPPI        & 16           & 201         & 12.30         & $-$0.048 \\
\bottomrule
\end{tabular}}
\end{table}

\begin{figure}[t]
  \centering
  \includegraphics[width=0.9\columnwidth]{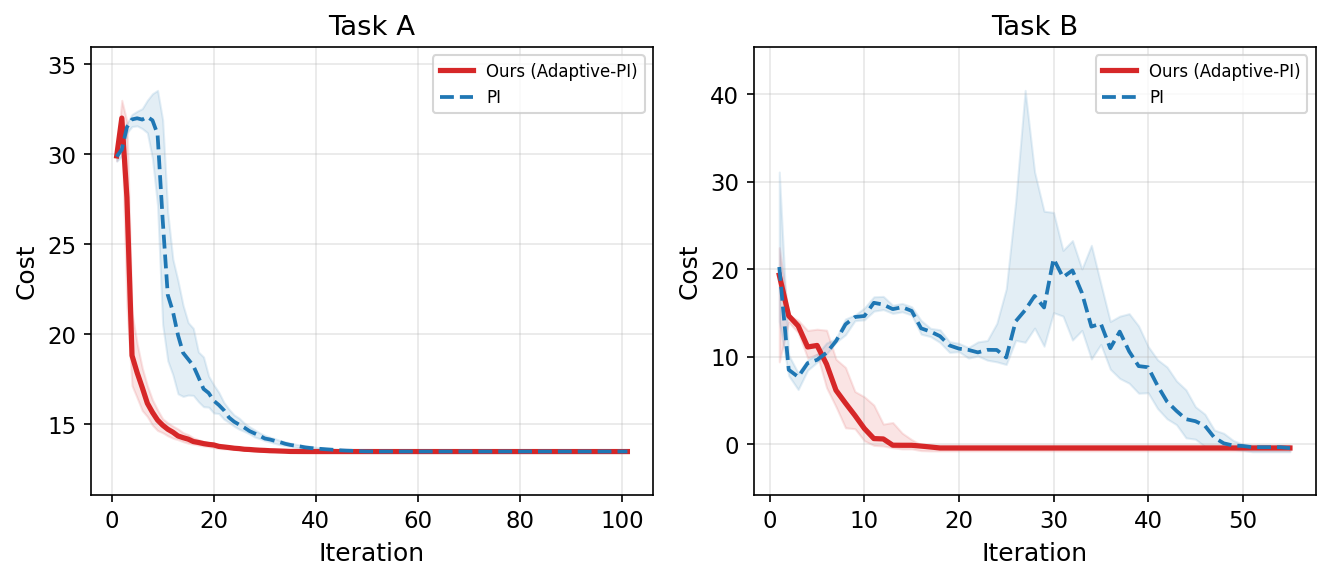}
  \caption{Cost convergence (median $\pm$ IQR, 100 runs) for all tasks. ITAC MPPI reaches the convergence at significantly fewer iterations with more stable performance.}
  \label{fig:conv}
\end{figure}
\vspace{0pt}
\begin{figure}[t]
  \centering
  \includegraphics[width=0.9\columnwidth]{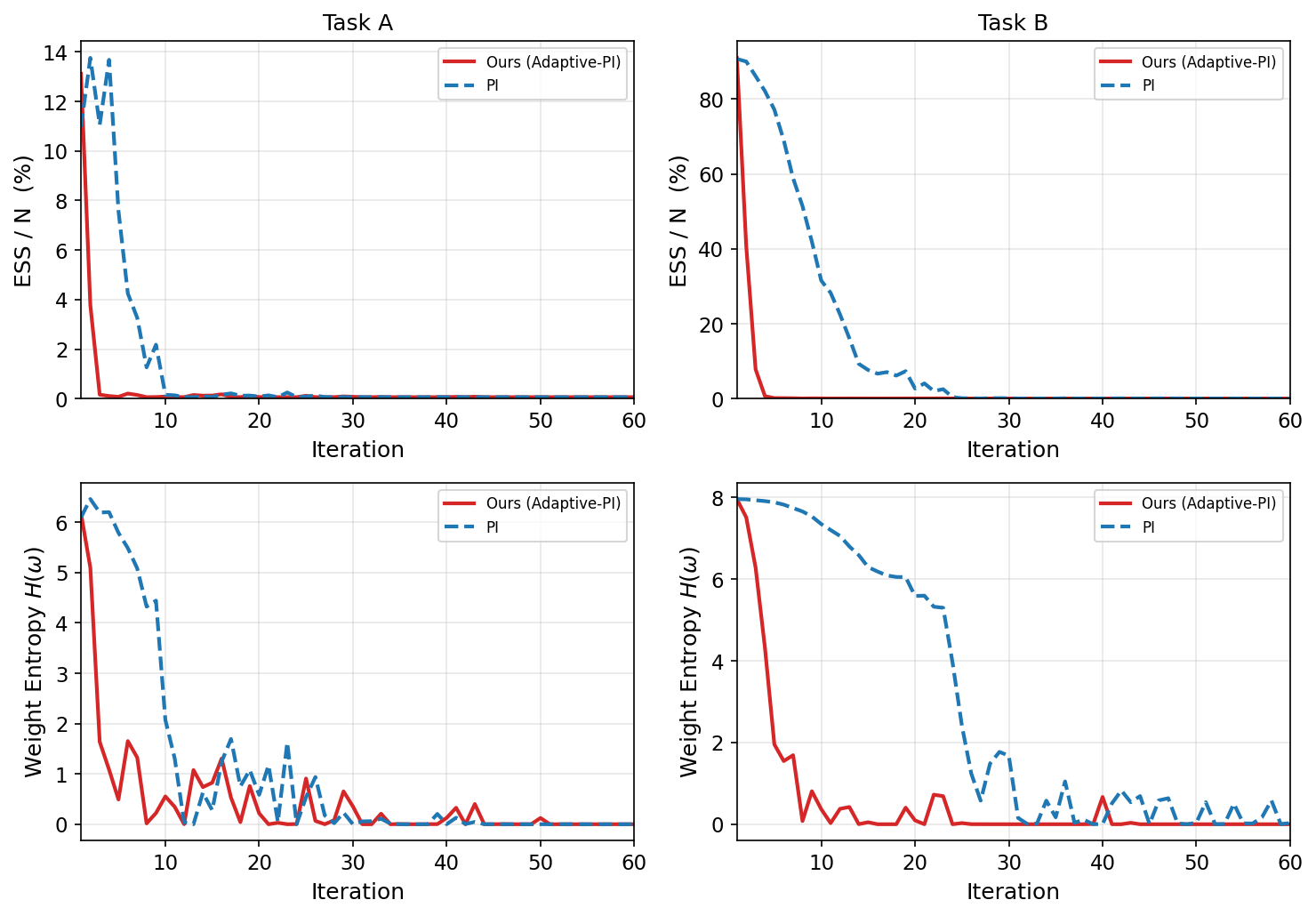}
  \caption{Weight entropy $H_M$ and ESS$/N$ vs.\ iteration.}
  \label{fig:entropy}
\end{figure}

While all methods achieve a 100\% success rate on simple STL Tasks~A and~B, ITAC substantially reduces the required iterations and computation time (by up to $4.14\times$). This improvement is mainly driven by the entropy-based cooling schedule, which adaptively lowers $\lambda$ as the sampling distribution concentrates. The ablation ITAC w.o.\ AG performs similarly on these two tasks, indicating that the acceleration primarily comes from entropy-driven cooling, while the anti-glassification mechanism offers only limited additional benefit in such simple settings.

Task~C (narrow corridor) is significantly more challenging, with all methods hitting the max iteration limit and exhibiting low success rates.
ITAC MPPI achieves 22\% SR compared to 18\% for ITAC w.o.\ AG and 16\% for Vanilla MPPI.
This result suggests that, in highly constrained environments, the anti-glassification mechanism can adaptively slow down the cooling process, thereby mitigating premature weight collapse and ultimately improving the probability of success.
 
Figure~\ref{fig:entropy} traces the weight distribution entropy $H$ and normalized ESS$/N$ for Tasks~A and~B. On Task~B, ITAC MPPI's ESS collapses from 91\% to below 0.1\% within 4 iterations, while Vanilla MPPI remains at 82\%. This rapid concentration of probability mass on high-quality trajectories directly explains the 73\% reduction in stopping iterations.
 
\vspace{-3pt}
\section{Conclusion} 
\vspace{-2pt}
This paper developed an information-theoretic adaptive cooling framework for deterministic MPPI, in which the temperature is regulated online using the Shannon entropy of the importance weights. The proposed mechanism provides guarantees on asymptotic convergence and yields a principled barrier against premature weight collapse. Experimental results on nonsmooth signal temporal logic tasks demonstrated improved sampling efficiency and faster convergence over existing baselines.

\bibliographystyle{IEEEtran}
\bibliography{ref}

@article{kappen2005linear,
  title={Linear theory for control of nonlinear stochastic systems},
  author={Kappen, Hilbert J.},
  journal={Physical Review Letters},
  volume={95},
  number={20},
  pages={200201},
  year={2005}
}

@article{theodorou2010generalized,
  title={A generalized path integral control approach to reinforcement learning},
  author={Theodorou, Evangelos and Buchli, Jonas and Schaal, Stefan},
  journal={Journal of Machine Learning Research},
  volume={11},
  pages={3137--3181},
  year={2010}
}

@article{kazim2024recent,
  title={Recent advances in path integral control for trajectory optimization: An overview in theoretical and algorithmic perspectives},
  author={Kazim, Muhammad and Hong, JunGee and Kim, Min-Gyeom and Kim, Kwang-Ki K.},
  journal={Annual Reviews in Control},
  volume={57},
  pages={100931},
  year={2024}
}

@inproceedings{yi2024covompc,
  title={{CoVO}-{MPC}: Theoretical analysis of sampling-based {MPC} and optimal covariance design},
  author={Yi, Zeji and Pan, Chaoyi and He, Guanqi and Qu, Guannan and Shi, Guanya},
  booktitle={6th Annual Learning for Dynamics \& Control Conference},
  pages={1122--1135},
  year={2024}
}

@inproceedings{asmar2023model,
  title={Model Predictive Optimized Path Integral Strategies},
  author={Asmar, Dylan M. and Senanayake, Ransalu and Manuel, Shawn and Kochenderfer, Mykel J.},
  booktitle={IEEE International Conference on Robotics and Automation (ICRA)},
  pages={3182--3188},
  year={2023}
}

@article{trevisan2024biased,
  title={Biased-{MPPI}: Informing sampling-based model predictive control by fusing ancillary controllers},
  author={Trevisan, Elia and Alonso-Mora, Javier},
  journal={IEEE Robotics and Automation Letters},
  volume={9},
  number={6},
  pages={5871--5878},
  year={2024}
}

@article{kappen2016adaptive,
  title={Adaptive importance sampling for control and inference},
  author={Kappen, Hilbert J. and Ruiz, Hans Christian},
  journal={Journal of Statistical Physics},
  volume={162},
  number={5},
  pages={1244--1266},
  year={2016}
}

@article{pezzato2025sampling,
  title={Sampling-based model predictive control leveraging parallelizable physics simulations},
  author={Pezzato, Corrado and Salmi, Chadi and Trevisan, Elia and Spahn, Max and Alonso-Mora, Javier and Hern\'{a}ndez Corbato, Carlos},
  journal={IEEE Robotics and Automation Letters},
  volume={10},
  number={3},
  pages={2750--2757},
  year={2025}
}

@article{fazlyab2026preconditioned,
  title={Model Predictive Path Integral Control as Preconditioned Gradient Descent},
  author={Fazlyab, Mahyar and Sharifi, Sina and Wang, Jiarui},
  journal={arXiv preprint arXiv:2603.24489},
  year={2026}
}

@inproceedings{watson2023inferring,
  title={Inferring smooth control: Monte carlo posterior policy iteration with gaussian processes},
  author={Watson, Joe and Peters, Jan},
  booktitle={Conference on Robot Learning},
  pages={67--79},
  year={2023}
}

@article{schramm2026variancereduced,
  title={Variance-Reduced Model Predictive Path Integral via Quadratic Model Approximation},
  author={Schramm, Fabian and Tiofack, Franki Nguimatsia and Perrin-Gilbert, Nicolas and Toussaint, Marc and Carpentier, Justin},
  journal={arXiv:2602.03639},
  year={2026}
}

@article{zheng2026signal,
  title={Signal Temporal Logic guided Stein Variational Path Integral Optimization},
  author={Zheng, Hongrui and Zang, Zirui and Amine, Ahmad and Vasile, Cristian Ioan and Mangharam, Rahul},
  journal={arXiv:2603.13333},
  year={2026}
}

@inproceedings{halder2025trajectoryplanningsignaltemporal,
  title={Trajectory Planning with Signal Temporal Logic Costs using Deterministic Path Integral Optimization},
  author={Halder, Patrick and Homburger, Hannes and Kiltz, Lothar and Reuter, Johannes and Althoff, Matthias},
  booktitle={IEEE International Conference on Robotics and Automation (ICRA)},
  year={2025}
}

@inproceedings{yoon2022samplingcomplexitypathintegral,
  title={Sampling Complexity of Path Integral Methods for Trajectory Optimization},
  author={Yoon, Hyung-Jin and Tao, Chuyuan and Kim, Hunmin and Hovakimyan, Naira and Voulgaris, Petros},
  booktitle={American Control Conference (ACC)},
  pages={4685--4690},
  year={2022}
}

@article{williams2017model,
  title={Model predictive path integral control: From theory to parallel computation},
  author={Williams, Grady and Aldrich, Andrew and Theodorou, Evangelos A},
  journal={Journal of Guidance, Control, and Dynamics},
  volume={40},
  number={2},
  pages={344--357},
  year={2017}
}

@article{homburger2025optimality,
  title={Optimality and suboptimality of MPPI control in stochastic and deterministic settings},
  author={Homburger, Hannes and Messerer, Florian and Diehl, Moritz and Reuter, Johannes},
  journal={IEEE Control Systems Letters},
  volume={9},
  pages={763--768},
  year={2025}
}

@article{kirkpatrick1983optimization,
  title={Optimization by simulated annealing},
  author={Kirkpatrick, Scott and Gelatt, C Daniel and Vecchi, Mario P},
  journal={Science},
  volume={220},
  number={4598},
  pages={671--680},
  year={1983}
}

@article{williamsInformationTheoreticModel2017a,
  title = {Information Theoretic Model Predictive Control: Theory and Applications to Autonomous Driving},
  author = {Williams, Grady and Drews, Paul and Goldfain, Brian and Rehg, James M. and Theodorou, Evangelos A.},
  journal = {IEEE Transactions on Robotics},
  volume = {34},
  number = {5},
  pages = {1141--1151},
  year = {2018}
}


\end{document}